\documentclass{ifacconf}
\usepackage{natbib}        
\usepackage{amsmath, amsfonts, amssymb, mathtools}
\usepackage{enumerate, graphicx, xcolor}
\usepackage{url}
\begin{document}
\begin{frontmatter}
   \title{Identification of a Kalman Filter: Consistency of Local Solutions}
   
   \author[First]{L{\'e}o Simpson} 
   \author[First]{and Moritz Diehl} 
   
   \address[First]{
      Department of Mathematics and
      Department of Microsystems Engineering (IMTEK),
      University of Freiburg, 79110 Freiburg, Germany
      (e-mails: \tt{\small firstname.familyname@imtek.uni-freiburg.de} )
   }
      
\begin{abstract}
Prediction error and maximum likelihood methods are powerful tools for identifying linear dynamical systems and, in particular, enable the joint estimation of model parameters and the Kalman filter used for state estimation.
A key limitation, however, is that these methods require solving a generally non-convex optimization problem to global optimality.
This paper analyzes the statistical behavior of local minimizers in the special case where only the Kalman gain is estimated.
We prove that these local solutions are statistically consistent estimates of the true Kalman gain.
This follows from asymptotic unimodality: as the dataset grows, the objective function converges to a limit with a unique local (and therefore global) minimizer.
We further provide guidelines for designing the optimization problem for Kalman filter tuning and discuss extensions to the joint estimation of additional linear parameters and noise covariances.
Finally, the theoretical results are illustrated using three examples of increasing complexity.
The main practical takeaway of this paper is that difficulties caused by local minimizers in system identification are, at least, not attributable to the tuning of the Kalman gain.

\end{abstract}

\begin{keyword}
   Linear systems,
   Linear system identification,
   Estimation and filtering,
   Kalman filtering,
   Optimization.
\end{keyword}
\end{frontmatter}

\newcommand{\cl}[1]{{\mathcal #1}}
\newcommand{\mr}[1]{{\mathrm{#1}}}
\newcommand{\R}{{\mathbb{R}}}
\newcommand{\E}{{\mathbb{E}}}
\newcommand{\N}{{\mathbb{N}}}

\newcommand{\GEQ}{{\, \succcurlyeq \, }}
\newcommand{\LEQ}{{\, \preccurlyeq \, }}
\newcommand{\GEQS}{{\, \succ \, }}
\newcommand{\LEQS}{{\, \prec \, }}

\newcommand{\Ker}[1]{\mr{Ker} \of{#1}}
\newcommand{\Trace}[1]{\mr{Tr} \of{#1}}
\newcommand{\abs}[1]{\left\lvert#1\right\rvert}
\newcommand{\norm}[1]{\left\lVert#1\right\rVert}
\newcommand{\mylimit}[1]{\xrightarrow[#1\to+\infty]{}}
\newcommand{\of}[1]{{\left(#1\right)}}
\newcommand{\ofc}[1]{{\left[#1\right]}}
\newcommand{\setdef}[2]{\left\{ #1 \; \big\vert \; #2 \right\}}
\newcommand{\clo}[1]{\mr{cl}\of{#1}}
\newcommand{\vect}[1]{\mr{vec}\of{#1}}
\newcommand{\logdet}[1]{\log\mr{det}\of{#1}}

\newcommand{\Eof}[1]{\E \left[#1\right]}
\newcommand{\sEof}[1]{\E [#1]}
\newcommand{\Covof}[1]{\mr{Cov}\left[#1\right]}
\newcommand{\Pof}[1]{\mathbb{P}\left[#1\right]}
\newcommand{\Gauss}[2]{\cl{N} \of{ #1, \, #2 }  }
\newcommand{\Gausstiny}[2]{\cl{N}(#1, \, #2)  }

\newcommand{\nx}{{n}}
\newcommand{\nuu}{{p}}
\newcommand{\ny}{{q}}
\newcommand{\nbeta}{{n_{\beta}}}

\newcommand{\changes}[1]{{\color{red} #1}}

\newcommand{\dK}{{D}}
\newcommand{\Set}{\cl{L}}
\newcommand{\K}{L}
\newcommand{\cUS}{\gamma}
\newcommand{\rhoUS}{\lambda}

\newcommand{\myparagraph}[1]{\vspace{-0.5cm}\paragraph*{#1} ~}
\newenvironment{proof}{\begin{pf}}{\vspace{-0.1cm} \\ \hspace*{\fill} \qed \end{pf} }
\newcommand\blfootnote[1]{
  \begingroup
  \renewcommand\thefootnote{}\footnote{#1}%
  \addtocounter{footnote}{-1}%
  \endgroup
}
\newcommand{\IMAGES}{images}

\newcommand{\transp}{^{\! \top} \!}  
\newcommand{\sminus}{\hspace{-0.07cm}-\hspace{-0.07cm}} 
\newcommand{\splus}{\!+\!} 
\newcommand{\sequal}{\!=\!} 


\section{Introduction}\label{section-introduction}

Identifying a model from measurements is an important task, especially for designing model-based controllers.
To efficiently apply such algorithms, three requirements are central: an accurate predictive model, efficient online state estimation, and, sometimes, uncertainty quantification.
Regarding the first requirement, it is often a mixture of prior knowledge from physics-based modeling and data-driven modeling.
A popular approach for this task is parametric system identification using Prediction Error Methods (PEM)~\citep{Ljung2002PEM} or Maximum Likelihood Estimation (MLE)~\citep{Astrom1979, Simpson2023}.
Regarding the second requirement, for Linear Time-Invariant (LTI) systems, online state estimation is often performed using Kalman filters~\citep{Anderson1979}.
Such a filter requires knowledge of the process and measurement noise covariance matrices, which are often difficult to derive from the system's physics.
Several approaches exist to estimate them from data~\citep{Abbeel2005, Odelson2006}, but if they are entirely unknown, it is often preferable to estimate the Kalman gain directly with PEM or MLE, possibly jointly with other parameters~\citep{Kuntz2025}.
This motivates analyzing PEM for Kalman gain estimation.
While this method has strong statistical guarantees, they require solving a generally non-convex optimization problem to global optimality for these guarantees to hold.
This is a limitation because derivative-based optimization algorithms can only guarantee convergence to a local minimizer.
A natural question arises: can we still provide statistical guarantees for local minimizers?

This paper provides a positive answer to this question for the case of Kalman gain estimation using PEM.
This follows from the fact that the optimization problem is asymptotically unimodal: as the amount of data goes to infinity, the limit of the objective function has a unique local (and therefore global) minimizer.
Thus, \emph{asymptotically, non-global local minimizers are not attributable to the estimation of the Kalman gain}.
While we do not analyze the finite-sample case, in our experience, a few samples are usually sufficient to observe the unimodality, as illustrated by numerical examples.
We also propose some extensions of this result to more general cases, such as the joint identification of the innovation covariance matrix with MLE.
However, one cannot provide guarantees for the completely general case because a poorly chosen parameterization can always lead to artificial local minima.

The asymptotic unimodality of PEM and MLE has been proven for a few specific classes of Single-Input Single-Output (SISO) systems.
A summary of classic results is given in~\citet[Section 10.5]{Ljung1999SystemId}.
A notable one is asymptotic unimodality for ARMA models~\citep{Astrom1974}, which are black-box single-output autonomous LTI systems.
Other results exist for SISO systems with specific structures or input design~\citep{Soderstrom1975,Goodwin2003,Zou2009,Eckhard2012}.
To the best of our knowledge, this paper provides the first asymptotic unimodality result for multiple-output systems (apart from the trivial case of linear regression), and the first time-domain analysis of this problem.

\myparagraph{Outline of this paper.}
In Section~\ref{section-problem}, we define the stochastic system and identification problem of interest.
In Section~\ref{section-method}, we formulate PEM for this problem, with stability-enforcing constraints, and recall the state-of-the-art result on the consistency of the global solutions.
In Section~\ref{section-consistency-local}, we state the main results of this paper: asymptotic unimodality and consistency of local minimizers.
These theoretical results are illustrated with numerical examples in Section~\ref{section-numerical}.
Finally, in Section~\ref{section-extension}, we discuss possible extensions of these results to more general settings.
Conclusions and future research directions are discussed in Section~\ref{section-conclusion}.

\myparagraph{Notation.}
Throughout this paper, we denote by $I$ the identity matrix of appropriate dimensions.
We use $\norm{x}_P \coloneqq \sqrt{x\transp P x}$ to denote the weighted norm associated with $P \GEQS 0$, and omit the subscript when $P = I$.
The trace of a matrix $M$ is denoted by $\Trace{M}$.
Finally, $\clo{\Set}$ denotes the closure of a set $\Set$.

\section{Problem Statement}\label{section-problem}
Consider the following system dynamics, in the innovation form, for $k=0, \dots, N$:
\begin{subequations}\label{eq-system}
\begin{align}
x_{k+1} &= A x_k + B u_k + \K e_k, \\
y_k &= C x_k + e_k.
\end{align}
\end{subequations}
Here, $x_k \in \R^{\nx}, u_k \in \R^{\nuu}, y_k \in \R^{\ny}$ and $e_k \in \R^{\ny}$ denote the state, the input, the output, and the innovation, respectively.
The goal is to estimate the observer gain $\K \in \Set  \subset \R^{\nx \times \ny} $ from the data $\{u_k, y_k\}_{k=0, \dots, N}$. \\
The key motivation for this innovation form is that, under mild assumptions, LTI systems with process and measurement noise can be equivalently represented in this form, with $\K$ being the Kalman gain.

\begin{assum}[Assumptions on the system]\label{assum-system}
~
\begin{enumerate}[a)]
    \itemsep0.1em
    \item The matrices $A, B, C$ and the initial state $x_0$ are known.
    \item The pair $(A, C)$ is observable.
    \item The matrix $C$ is full row-rank, i.e., $\mr{rank}(C) = \ny$.
    \item The data are generated by~\eqref{eq-system} with $\K = \K^\star \in \Set$.
    \item The matrix $A - \K^\star C$ is stable, i.e. $\rho(A - \K^\star C) < 1$.
    \item The innovations $e_k$ are zero-mean independent random variables such that:
    \begin{align}
        \Eof{e_k e_k^\top}=S^\star\GEQS 0, \quad \text{and} \quad
        \sup_{k \in \N} \Eof{\norm{e_k}^4} < \infty.
    \end{align}
\end{enumerate}
\end{assum}
The covariance matrix $S^\star \GEQS 0$ is, in general, unknown, but we do not attempt to estimate it (apart from Section~\ref{section-extension}).
Note that assumption e) is not restrictive as the Kalman gain $\K^\star$ is always stabilizing the system when $(A, C)$ is observable.

\section{The prediction error method}\label{section-method}
In the presented settings, the PEM~\citep{Ljung2002PEM} leads to the following optimization problem:
\begin{align}\label{eq-PEM}
    \underset{\K \in \Set}{\text{minimize}} \quad  V_N(\K) \coloneqq \frac{1}{N} \sum_{k=1}^N \norm{y_k - C \hat{x}_k(\K)}^2_{W}
\end{align}
where the weighting matrix $W \in \R^{\ny \times \ny}$ is some positive-definite matrix, and the predicted states $\hat{x}_k(\K)$ are obtained from the Kalman filter equations:
\begin{align}\label{eq-kalman}
    \hat{x}_{k+1}(\K) &= A \hat{x}_k(\K) + B u_k + \K \of{y_k - C \hat{x}_k(\K)}.
\end{align}
A crucial observation is that the predicted states $\hat{x}_k(\K)$ depend nonlinearly on the parameter $\K$ because of the term ``$\K C \hat{x}_k(\K)$'' in~\eqref{eq-kalman}.
In fact, this is what makes the optimization problem~\eqref{eq-PEM} non-convex in general.
Now we make an important assumption regarding the feasible set $\Set$.
\begin{assum}[Uniform stability]\label{assum-K}
    The family of matrices $\setdef{A - \K C}{\K \in \Set}$ are uniformly stable, i.e., for some constants $\cUS > 0$ and $\rhoUS \in (0,1)$, we have:
    \begin{align}\label{eq-uniform-stability}
        \forall \K \in \Set, \quad \forall i \in \N, \quad \norm{(A - \K C)^i} \leq \cUS \rhoUS^i.
    \end{align}
\end{assum}
Importantly, $\Set$ is necessarily bounded because~\eqref{eq-uniform-stability} implies $\norm{\K C} \leq \norm{A} + \cUS$ and $C$ is full row-rank (cf. Assumption~\ref{assum-system}).
Therefore, $\clo{\Set}$ is compact, which will be helpful later.
By a continuity argument, \eqref{eq-uniform-stability} also holds for all $\K \in \clo{\Set}$.

From a more practical point of view, there exist several ways to impose such a uniform stability constraint.
For example, in~\citet{Kuntz2025}, this is done via Linear Matrix Inequalities (LMI).
Similarly, we propose to ensure~\eqref{eq-uniform-stability} using some Lyapunov equation constraint, as it is done in~\citet{Diehl2009a}:
\begin{subequations}\label{eq-PEM-stability}
    \begin{align}
        \underset{\K \in \R^{\nx \times \ny}, P \in \R^{\nx \times \nx} }{\text{minimize}} \;  &\frac{1}{N} \sum_{k=1}^N \norm{y_k - C \hat{x}_k(\K)}^2_{W}, \\
        \text{subject to} \quad
        &P = (A - \K C) P (A -  \K C)\transp + I, \label{constraint-Lyapunov} \\
        &  \alpha \Trace{P-I} \leq 1, \label{constraint-Trace} \\
        & P \GEQ 0, \label{constraint-PSD}
    \end{align}
\end{subequations}
for some $\alpha > 0$.
In the our experience, this method is quite robust against the choice of parameter $\alpha$: it only needs to be small enough for $\K^\star$ to be feasible in~\eqref{eq-PEM-stability}.
We choose $\alpha = 0.02$ in the numerical examples.

The constraints~\eqref{constraint-Lyapunov}-\eqref{constraint-PSD} are equivalent to $\rho_{\alpha}(A-\K C) \leq 1$ where $\rho_{\alpha}(\cdot)$ is the smooth spectral radius approximation used in~\citet{Diehl2009a}.
As noted there, the constraint $P \GEQ 0$ is in fact never active (because $P \GEQ I \GEQS 0$ for any feasible point), so even though~\eqref{eq-PEM-stability} is a nonlinear semi-definite program, it can be (almost) treated as an ordinary nonlinear program in practice.

Proposition~\ref{prop-K-stability} below draws a connection between~\eqref{eq-PEM-stability} and Assumption~\ref{assum-K}.
A proof is provided in the Appendix. 
\begin{prop}\label{prop-K-stability}
    For any $\alpha > 0$, the set $\Set_{\alpha}$ defined below satisfies Assumption~\ref{assum-K}:
    \begin{align}\label{eq-K-stability}
        \Set_{\alpha}  \sequal
        \bigg\{
            \K \!\in \R^{\nx \times \ny} \text{ s.t.\ (\ref{constraint-Lyapunov}-\ref{constraint-Trace}) holds for some } P \GEQ 0
        \! \bigg\}.
    \end{align}
\end{prop}

\myparagraph{Consistency of the global solution.}
It is well known that the PEM is strongly consistent~\citep[Theorem 8.2]{Ljung1999SystemId}: the global minimizers of~\eqref{eq-PEM} converge almost surely (i.e., with probability one) to the true parameters $\K^\star$ when $N$ goes to infinity.
The proof relies on the fact that, almost surely, the objective function $V_N(\K)$ converges uniformly to its expected value~\citep[Lemma 8.2]{Ljung1999SystemId}, and that $\K^\star$ minimizes this expected value:
\begin{align}\label{eq-minimizer:EV}
    \Eof{V_N(\K)}\sequal
    \Trace{S^\star W} \splus \underbrace{\Eof{\norm{C \of{\hat{x}_k(\K) - \hat{x}_k(\K^\star)}}^2_{W}}}_{
        \textup{minimized for } \K = \K^\star.
    }
\end{align}
Regarding the uniform convergence of $V_N(\K)$, the proof relies on a lemma for stochastic dynamical systems presented in~\citet[Theorem 2B.1]{Ljung1999SystemId} and repeated in the Appendix of the present paper.

\section{Consistency of local solutions}\label{section-consistency-local}
In this section, we prove that strong consistency also holds for local minimizers of~\eqref{eq-PEM} that are in the interior of $\Set$.
This relies on two main results: the first provides the limit of the objective function and its derivatives, and the second establishes the unimodality of this limit.

Before stating Lemma~\ref{lem-uniform-convergence}, we make some important simplifications of the function $V_N(\K)$:
\begin{align}\label{eq:VN-simplified}
    V_N(\K) &= \frac{1}{N} \sum_{k=1}^N \norm{e_k - C z_k(\K) }^2_{W},
\end{align}
where $z_k(\K) \coloneqq \hat{x}_k(\K) - \hat{x}_k(\K^\star)$ can be computed recursively from $z_0(\K) = 0$ and:
\begin{align}\label{eq:zk}
    z_{k+1}(\K) &= (A - \K C) z_k(\K) + (\K - \K^\star) e_k.
\end{align}
Interestingly, $V_N(\K)$ does not depend on the inputs $u_k$.

Now we define the steady-state error covariance $\bar{\Sigma}(\K)$ as the unique solution of the following Lyapunov equation:
\begin{align}\label{eq-def-Sigma-bar}
    \bar{\Sigma}(\K) &= (A \sminus \K C) \bar{\Sigma}(\K) (A \sminus \K C)\transp \splus (\K \sminus \K^\star\!) S^\star \! (\K \sminus \K^\star \!)\transp.
\!\!
\end{align}
This allows us to define the function $\bar{V}(\K)$ as:
\begin{align}\label{eq-def-V-bar}
    \bar{V}(\K) &\coloneqq  \Trace{W \of{ S^\star + C \bar{\Sigma}(\K) C\transp} },
\end{align}
Intuitively, $\bar{V}(\K)$ represents the expected value of the objective where we replaced the error covariances with the steady-state solution of the corresponding Lyapunov equation.
The following lemma states the convergence of $V_N(\K)$ and its gradient to $\bar{V}(\K)$.
\begin{lem}[Uniform convergence of the gradients]\label{lem-uniform-convergence}
    ~
    The gradient of $V_N(\cdot)$ converges almost surely and uniformly to the gradient of $\bar{V}(\cdot)$:
    \begin{align}\label{eq-uniform-convergence-dV}
        \Pof{\sup_{\K \in \Set} \norm{ \nabla V_N(\K) - \nabla \bar{V}(\K)} \mylimit{N} 0} = 1.
    \end{align}
\end{lem}
The proof is provided in the Appendix.
Note that the value of $V_N(\K)$ also converges uniformly to $\bar{V}(\K)$ over $\K \in \Set$, but we do not need this assertion here.

We are now ready to state our main result:
\begin{thm}[Unimodality of the limit]\label{thm-unimodality}
    The unique stationary point of $\bar{V}(\K)$ in $\clo{\Set}$ is $\K^\star$:
    \begin{align}
        \K \in \clo{\Set} \text{ and } \nabla \bar{V}(\K) = 0  \quad \iff \quad \K = \K^\star.
    \end{align}
\end{thm}
\begin{proof} ``$\Leftarrow$'':
    Since $\bar{V}(\K^\star) = \Trace{W S^\star} = \min_\K \bar{V}(\K)$ and $\K^\star$ is in the interior of $\Set$, it is clear that $\K^\star$ is a stationary point of $\bar{V}(\K)$.

    ``$\Rightarrow$'':
    Let $\hat{\K} \in \clo{\Set}$ be such that $\nabla \bar{V}(\hat{\K}) = 0$.
    Define the direction $\dK \coloneqq (\hat{\K} - \K^\star) S^\star - (A-\hat{\K}C)\bar{\Sigma}(\hat{\K})C\transp$, and the directional derivative $\dot{\Sigma}$ in that direction, i.e.:
    \begin{align}\label{eq-proof-unimodality-1}
        \dot{\Sigma} = \lim_{\varepsilon \to 0} \frac{\bar{\Sigma}(\hat{\K} + \varepsilon \dK) - \bar{\Sigma}(\hat{\K})}{\varepsilon}.
    \end{align}
    By differentiating~\eqref{eq-def-Sigma-bar} in the direction $\dK$, we find:
    \begin{align}\label{eq-proof-unimodality-2}
        \dot{\Sigma} = (A - \hat{\K} C) \dot{\Sigma} (A - \hat{\K} C)\transp + 2 \dK \dK\transp,
    \end{align}
    which itself implies:
    \begin{align}\label{eq-proof-unimodality-2.5}
        \!C\dot{\Sigma}C\transp \sequal 2\!  \sum_{i=0}^{+\infty} \!\of{C(A - \hat{\K} C)^i \dK} \of{C(A - \hat{\K} C)^i \dK}\transp\GEQ 0.\!
    \end{align}
    On the other hand, from the stationarity of $\hat{\K}$, we have $\Trace{W C \dot{\Sigma} C\transp } = 0$,
    which implies that $C \dot{\Sigma} C\transp = 0$ because $W \GEQS 0$ and $C \dot{\Sigma} C\transp \GEQ 0$.
    Since all of the terms of the zero-sum~\eqref{eq-proof-unimodality-2.5} are positive semi-definite, we deduce that each term is zero, i.e.
    \begin{align}\label{eq-proof-unimodality-3}
        \forall i \in \N, \quad C(A - \hat{\K} C)^i \dK = 0.
    \end{align}
    Since the pair $[A, C]$ is observable (cf. Assumption~\ref{assum-system}.a), the pair $[A - \hat{\K} C, C]$ is also observable (this is a consequence of the Hautus lemma~\cite{Hautus1969}).
    Hence, \eqref{eq-proof-unimodality-3} implies that $\dK = 0$.
    We continue as follows:
    \begin{subequations}\label{eq-proof-unimodality-4}
        \begin{align}
            0 &= \dK (\hat{\K} - \K^\star)\transp \\
            &=\! (\hat{\K} \sminus \K^\star) S^\star (\hat{\K} \sminus \K^\star)\transp \sminus (A \sminus \hat{\K}C)\bar{\Sigma}(\hat{\K})C\transp (\hat{\K} \sminus \K^\star)\transp \!\!\!\!\! \\
            &= \bar{\Sigma}(\hat{\K})  - (A-\hat{\K}C)\bar{\Sigma}(\hat{\K})(A-\K^\star C)\transp, \label{subeq-proof-unimodality}
        \end{align}
    \end{subequations}
    where we used again equation~\eqref{eq-proof-unimodality-2} to get~\eqref{subeq-proof-unimodality}.
    Repeating $i$ times the equality induced by~\eqref{subeq-proof-unimodality}, we find:
    \begin{align}\label{eq-proof-unimodality-5}
        \bar{\Sigma}(\hat{\K}) &= (A-\hat{\K}C)^i \bar{\Sigma}(\hat{\K}) (A-\K^\star C)^{i\top}.
    \end{align}
    Taking the limit when $i \to +\infty$, we find that $\bar{\Sigma}(\hat{\K})=0$.
    This limit holds because of the stability of the matrices $A - \K C$ for any $\K \in \clo{\Set}$ (cf. Assumption~\ref{assum-K} and the remark after it).
    Finally, injecting $\bar{\Sigma}(\hat{\K})=0$ in~\eqref{eq-def-Sigma-bar} and using the fact that $S^\star$ is positive-definite (cf. Assumption~\ref{assum-system}) we find $\hat{\K} = \K^\star$, which concludes the proof.
\end{proof}
Combining Lemma~\ref{lem-uniform-convergence} and Theorem~\ref{thm-unimodality} yields the following consistency result for local minimizers:
\begin{thm}[Consistency of stationary points]\label{thm-consistency}   
    ~
    Let, for all $N \in \N$, $\hat{\K}_N \in \Set$ be stationary points of $V_N(\cdot)$. \\
    Then, $\hat{\K}_N$ is a strongly consistent estimate of $\K^\star$:
    \begin{align}
        \Pof{ \hat{\K}_N \mylimit{N} \K^\star } = 1.
    \end{align}
\end{thm}
\begin{proof}
    Consider a realization for which the convergence from Lemma~\ref{lem-uniform-convergence} holds.
    Thus, for this realization, $\nabla \bar{V}(\hat{\K}_N)$ converges to zero.
    Furthermore, the sequence $\{\hat{\K}_N\}_{N \in \N}$ lies in  $\clo{\Set}$, which is compact (see the remark after Assumption~\ref{assum-K}).
    Let $\bar{\K}$ be any limit point of this sequence.
    By continuity of $\nabla \bar{V}(\cdot)$, we have $\nabla \bar{V}(\bar{\K}) = 0$.
    Using Theorem~\ref{thm-unimodality}, this implies that $\bar{\K} = \K^\star$.
    Since this holds for any limit point of the sequence $\{\hat{\K}_N\}_{N \in \N}$, this sequence converges to $\K^\star$ (for this realization).
    This is true for any realization in a probability-one set, and the desired almost sure convergence follows.
\end{proof}
\begin{cor}[Consistency of local minimizers]
    If $\hat{\K}_N$ are local minimizers of $V_N(\cdot)$ in the interior of $\Set$, then they are strongly consistent estimates of $\K^\star$.
\end{cor}
\begin{proof}
    This is a direct consequence of Theorem~\ref{thm-consistency} because local minimizers in the interior of the feasible set are stationary points.
\end{proof}

The results of this section rely on Assumptions~\ref{assum-system} and~\ref{assum-K}.
Even if a feasible set $\Set$ satisfying Assumption~\ref{assum-K} is not explicitly used, Theorem~\ref{thm-consistency} still holds provided that the sequence $\{\hat{\K}_N\}_{N \in \N}$ satisfies the uniform stability condition~\eqref{eq-uniform-stability}, for some $\gamma > 0$ and $\lambda \in (0,1)$.
This is the case whenever $\hat{\K}_N$ are bounded, and $\rho(A - \hat{\K}_N C)$ is bounded away from $1$.
As a conjecture, we also claim (without a proof) that this condition holds whenever the objective value $V_N(\hat{\K}_N)$ remains bounded.

\section{Numerical Examples}\label{section-numerical}

In this section, we illustrate the results of Section~\ref{section-consistency-local} with three examples.
First, a one-dimensional toy system allows us to visualize the objective function and its limit.
Second, a two-state system reveals the optimization landscape when several initial guesses are used.
Finally, a more realistic multi-output system highlights the consistency of the estimates.
All experiments are reproducible using the code accompanying this paper\footnotemark.

\myparagraph{A one-dimensional illustrative example.}
To visualize asymptotic unimodality, we consider a single-state system with known scalars $A, C \in \R$, and we generate data from~\eqref{eq-system} with a scalar Kalman gain $\K^\star \in \R$ and Gaussian innovations.
The objective function $V_N(\K)$ in~\eqref{eq-PEM} is evaluated for different values of $N$ using the weighting $W = 1$.
\begin{figure}[ht]
    \vspace{-0.2cm}
	\begin{center}
        \includegraphics[width=\linewidth]{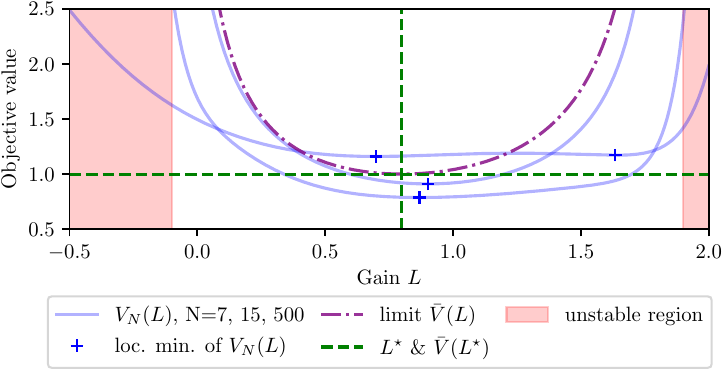}
	\end{center}
	\caption{
    Objective $V_N(\K)$ and its limit $\bar{V}(\K)$ for a one-dimensional system with $A = 0.9$, $C = 1, \K^\star = 0.8$ and $e_k \sim \Gausstiny{0}{1}$.
    }\label{fig-objective-function}
 \vspace{-0.3cm}
\end{figure}

Figure~\ref{fig-objective-function} shows $V_N(\K)$ for three values of $N$, together with its limit $\bar{V}(\K)$.
The shaded region indicates gains $\K$ for which $A-\K C$ is unstable (i.e. $\abs{A-\K C}\geq 1$ here).
For small $N$, several local minima are visible, even in this simple setting.
As $N$ increases, these spurious minima disappear and $V_N(\K)$ becomes unimodal, with its minimizer approaching the true gain $\K^\star$.

\myparagraph{A two-state example with scalar measurements.}
We next study a system with two states and scalar measurements.
The underlying continuous-time dynamics describe a particle subject to linear friction and a random piecewise-constant force:
\begin{align}\label{eq-dynamics-particle}
    \ddot{q}(t) &= -\mu \dot{q}(t) + w_k, \quad t \in [k\Delta t, (k+1)\Delta t),
\end{align}
where $q(t)$ is the position of the particle at time $t$, $\mu > 0$ is the friction coefficient, and $w_k \sim \Gausstiny{0}{\sigma_w^2}$ is a random force that remains constant over each sampling period of length $\Delta t > 0$.
The measurements take the form $y_k = q(k \Delta t) + v_k$, where $v_k \sim \Gausstiny{0}{\sigma_v^2}$ is the measurement noise.
This system is discretized analytically into a discrete-time LTI model with $x_k = [q(k \Delta t), \dot{q}(k \Delta t)]\transp$, which we can put in the innovation form~\eqref{eq-system} after computing the true Kalman gain $\K^\star$ from the Discrete Algebraic Riccati Equation (DARE)~\citep{Anderson1979}.
To estimate the gain, we solve the constrained PEM problem~\eqref{eq-PEM-stability} with the Lyapunov-based stability constraint with some constant $\alpha \in (0,1)$.
In practice, we optimize over $L$ only by eliminating $P$ via~\eqref{constraint-Lyapunov}.
Derivatives are computed with CasADi, and an interior-point method~\citep{Nocedal2006} is implemented\footnotemark[\value{footnote}] with line-search and a Gauss-Newton Hessian approximation.
\footnotetext{
    available at \url{https://github.com/Leo-Simpson/KalmanId}.
}

To explore the optimization landscape, we draw $50$ initial guesses uniformly in the feasible set $\Set_{\alpha}$ and solve the problem for different data lengths $N$.
Figure~\ref{fig-opti} shows the iterates and final solutions in the plane $(L_{11}, L_{21})$ together with the stable region where $\rho(A-LC)<1$ and the associated subset where $\alpha \Trace{P(L)-I}\leq 1$.
Note that the latter acts as a barrier for the former, which is consistent with the theory in~\citet{Diehl2009a}.
\begin{figure}[ht]
    \vspace{-0.1cm}
	\begin{center}
        \includegraphics[width=\linewidth]{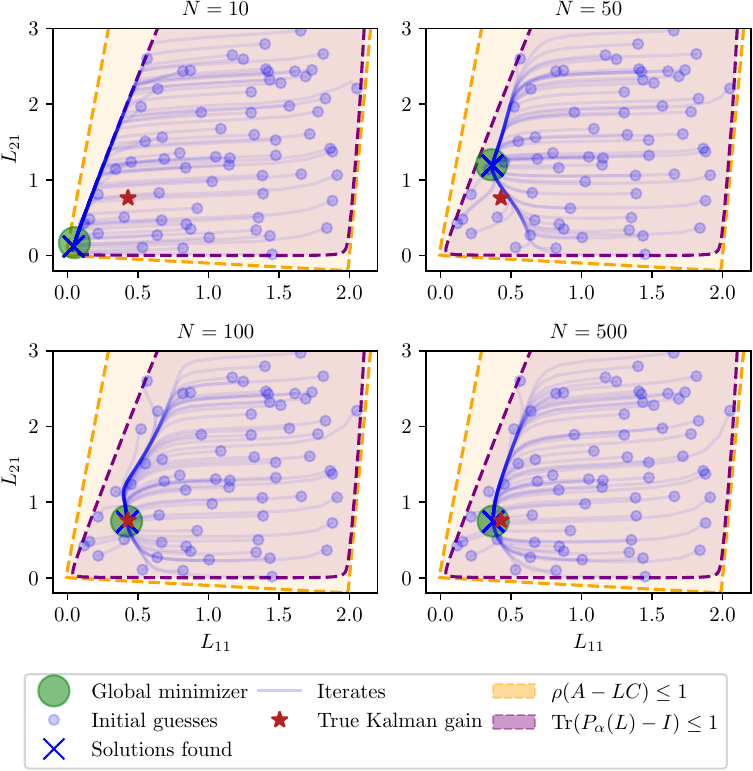}
	\end{center}
	\caption{
        Optimization iterates and solutions for the two-state example with $\alpha=0.02$,
        and the model parameters $\mu = 0.1$, $\Delta t = 0.1$, $\sigma_w = 10$, and $\sigma_v = 1$.
    }\label{fig-opti}
 \vspace{-0.3cm}
\end{figure}

Across all initializations and values of $N$, the algorithm always converges to the same point, which coincides with the global solution found by a dense grid search over $\Set_{\alpha}$.
As predicted by Theorem~\ref{thm-consistency}, this solution approaches the true gain $\K^\star$ as $N$ increases.

\myparagraph{A two-state example with two measurements.}
Finally, we consider a slightly more realistic problem: the same model as~\eqref{eq-dynamics-particle} but where both position and acceleration are being measured:
\begin{align}
    y_k = \begin{bmatrix} q(k \Delta t) + v_k & \qquad \ddot{q}(k \Delta t)  + v'_k \end{bmatrix}^\top.
\end{align}
Measurement noises are still Gaussian:  $v_k \sim \Gausstiny{0}{\sigma_{v}^2}$ and $v'_k \sim \Gausstiny{0}{\sigma_{v'}^2}$.
The process noise $w_k$ follows a mixture distribution: $w_k \sim \Gausstiny{0}{\sigma_w^2}$ with probability $p > 0$, and $w_k = 0$ with probability $1-p$.
Note that $w_k$ still has zero mean, finite fourth-order moments, and a positive variance $\Eof{w_k^2} = p \sigma_w^2$.

We generate several independent realizations of the dataset and, for each realization and each data length $N$, compute the gain estimate $\hat{L}_N$ using the same PEM formulation and optimization setup as in the previous example.
Figure~\ref{fig-consistency} compares $\hat{L}_N$ with the true Kalman gain $L^\star$, computed with the DARE as before.
\begin{figure}[ht]
	\begin{center}
        \includegraphics[width=\linewidth]{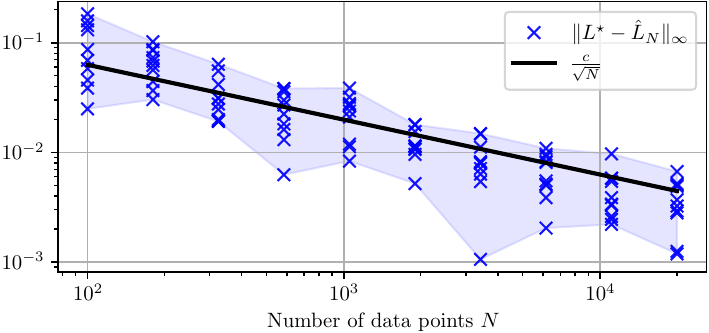}
	\end{center}
	\caption{
        Estimation error in the example with two measurements, for $10$ realizations.
        The model parameters are as in Figure~\ref{fig-opti}, and the noise parameters are $\sigma_{v} = 2$, $\sigma_{v'} = 1$, $\sigma_{w} = 1$ and $p = 0.1$.
    }\label{fig-consistency}
 \vspace{-0.1cm}
\end{figure}
The results show a clear decrease in the estimation error as $N$ grows, despite the non-Gaussian process noise.
This empirical behavior is consistent with the strong consistency of local minimizers established in Theorem~\ref{thm-consistency}.
We also observe a convergence speed of order $\mathcal{O}(1/\sqrt{N})$ which is consistent with the law of large numbers that was used in the proof.

Finally, for $N=100$, we solve the problem starting from $50$ different initial guesses generated using Kalman filters with random matrices $Q$ and $R$.
We observe that all distances between the solutions are less than $10^{-4}$, which is the numerical tolerance of the solver.
This suggests that all initializations lead to the same minimizer, which provides empirical evidence of the asymptotic unimodality of the cost.

\section{Extensions}\label{section-extension}

It is possible to extend the results established in Section~\ref{section-consistency-local} to the case where the covariance matrix $S=\Covof{e_k}$ is also estimated jointly with $\K$.
One can even extend the results by estimating additional parameters that enter linearly in the system dynamics.
Consider the following system dynamics, in the innovation form, for $k=1, \dots, N$:
\vspace{-0.25cm}
\begin{subequations}\label{eq-system_extended}
\begin{align}
x_{k+1} &= A x_k + B u_k + \Phi_k \beta + \K e_k, \\
y_k &= C x_k + e_k.
\end{align}
\end{subequations}
where $\beta \in \R^{\nbeta}$ is an unknown parameter vector to be estimated, and $\Phi_k \in \R^{\nx \times \nbeta}$ are some known regression matrices.
The parameters to be estimated are $\theta \coloneqq (\beta, \K, S)$.
Here, the predicted states also depend on $\beta$:
\begin{align}\label{eq-kalman_extended}
    \hat{x}_{k+1}(\theta) &= A \hat{x}_k(\theta) + B u_k + \Phi_k \beta
    + \K \of{y_k - C \hat{x}_k(\theta)}.
\end{align}
Regarding the joint estimation of $S$, the cost function in~\eqref{eq-PEM} must also be modified.
Instead, the MLE approach for dynamical systems can be employed~\citep{Astrom1979, Simpson2023}:
\begin{align}\label{eq-MLE}
    \underset{\theta = (\beta, \K, S)}{\text{minimize}} \;
        \frac{1}{N} \sum_{k=1}^N \norm{y_k - C \hat{x}_k(\theta)}^2_{S^{-1}} \splus \logdet{S}
\end{align}
As in~\eqref{eq:VN-simplified}, this cost function can be rewritten as follows:
\begin{align}\label{eq:MLE-simplified}
     J_N(\beta, \K, S) = \frac{1}{N}& \sum_{k=1}^N \norm{e_k - C z_k(\K) + \Psi_k(\K)(\beta^\star - \beta) }^2_{S^{-1}}   \nonumber \\
     &\hspace{15mm} + \logdet{S}, 
\end{align}
with $\Psi_k(\K) \coloneqq C \sum_{i=1}^{k} (A-\K C)^{i-1} \Phi_{k-i}$, and $z_k(\K)$ defined as in~\eqref{eq:zk}.
Here, to ensure consistency of the estimates, we also need some Persistent Excitation (PE) condition for $N$ large enough and some stabilizing gain $\bar{\K}$:
\begin{align}\label{eq-PE-condition}
    \frac{1}{N} \sum_{k=1}^N \Psi_k(\bar{\K})\transp\Psi_k(\bar{\K}) \GEQ c I \GEQS 0.
\end{align}
One can prove, using some linear algebra, that when this condition holds for some $\bar{\K}$, it also holds for any other $\K$.

Furthermore, to prove asymptotic unimodality, we also need to assume that the noise $e_k$ is independent of the regressors $\Phi_k$, which  excludes the case $\Phi_k=\Phi(u_k)$ when the data come from a closed-loop experiment unfortunately.
\begin{claim}
    The results from Section~\ref{section-consistency-local} still hold for the optimization problem~\eqref{eq-MLE}, under the assumption given above.
\end{claim}
\begin{proof}[Sketch of proof]
    Using~\eqref{eq:MLE-simplified}, the decomposition  $\Eof{J_N(\theta)} = \bar{J}^{\beta}_N(\beta, \K, S^{-1}) + \bar{J}^{\mr{L}}_N(\K, S^{-1}) + \bar{J}^{\mr{S}}_N(S)$ holds true, with the following terms:
    \begin{subequations}
        \begin{align}
            \bar{J}^{\beta}_N(\beta, \K, W) &\coloneqq  \norm{\beta - \beta^\star}^2_{Q_N(\K,W)}, \\
            \bar{J}^{\mr{L}}_N(\K, W) &\coloneqq  \frac{1}{N} \sum_{k=1}^N \Eof{\norm{C z_k(\K)}^2_{W}}, \\
            \bar{J}^{\mr{S}}_N(S) &\coloneqq  \Trace{S^{-1} S^\star } + \logdet{S},
        \end{align}
    \end{subequations}
    with $Q_N(\K,W) \coloneqq \frac{1}{N} \sum_{k=1}^N \Psi_k(\K)\transp W \Psi_k(\K)$. \\
    Let $\hat{\theta} = (\hat{\beta}, \hat{\K}, \hat{S})$ be a stationary point of $\sEof{J_N(\theta)}$.
    Thus, $\hat{\beta}$ is a stationary point of $\bar{J}^{\beta}_N(\beta, \K, W)$ for $W = \hat{S}^{-1}$, which implies $Q_N(\hat{\K}, W) (\hat{\beta} - \beta^\star) =0$.
    Using the PE condition~\eqref{eq-PE-condition}, we have $Q_N(\hat{\K}, W) \GEQS 0$, so $\hat{\beta} = \beta^\star$.
    Then, $\hat{\K}$ is a stationary point of $\bar{J}^{\mr{L}}_N(\K, W)$.
    Note that this function is the same as $\sEof{V_N(\K)}$ from Section~\ref{section-consistency-local}.
    Therefore, asymptotically its unique stationary point is $\K^\star$, provided that $S$ remains positive-definite and bounded.
    These first steps imply $\sEof{J_N(\hat{\theta} )}=\bar{J}^{\mr{S}}_N(\hat{S})$, and the unique stationary point of $\bar{J}^{\mr{S}}_N(S)$ is $S^\star$, so $\hat{S} = S^\star$, which concludes the sketch of the proof.
\end{proof}

\section{Conclusion}\label{section-conclusion}

This paper established that, when estimating only the Kalman gain of a known linear system, the PEM objective becomes asymptotically unimodal, and therefore remains a reliable identification method even when only local optimality can be guaranteed.
The numerical examples support this conclusion: spurious minimizers disappear as the dataset grows, and standard optimization routines consistently converge to the true gain.
The sketched extension to jointly estimating additional linear parameters and noise covariances is also encouraging, and formal proofs for such extensions are a future research direction.

Several open questions remain, for example, regarding the influence of model mismatch or the convergence speed of the local optima.
Another promising direction is the development of optimization algorithms tailored to this problem class; for instance, alternating or sequential updates over model parameters, Kalman gains, and covariances may yield fast and guaranteed convergence.


\bibliography{biblio}
\appendix

\section{Technical proofs}
\begin{proof}[Proof of Proposition~\ref{prop-K-stability}]
    Since $A - \K^\star C$ is stable (cf. Assumption~\ref{assum-system}), the constraints (\ref{constraint-Lyapunov}-\ref{constraint-PSD}) are satisfied for some $P^\star$ if $\K = \K^\star$ and if $\alpha$ is small enough.
    This is a consequence of the fact that a discrete-time Lyapunov equation always has a solution when the corresponding matrix is stable.

    Now let $\K \in \Set_{\alpha}$, and let $P$ be a matrix satisfying (\ref{constraint-Lyapunov}-\ref{constraint-PSD}).
    Note that $\alpha (P-I) \LEQ I$.
    This leads to:
    \begin{align}
        \! P \GEQ (1+\alpha) (P-I) = (1+\alpha) (A \sminus \K C) P (A \sminus \K C)\transp.
    \end{align}
    Iterating this inequality $i$ times leads to:
    \begin{align}
        P \GEQ (1+\alpha)^i (A \sminus \K C)^i P (A \sminus \K C)^{i \top}.
    \end{align}
    Using $I \LEQ P \LEQ \of{1+\alpha^{-1}} I$, we can obtain:
    \begin{align}
        (1+\alpha)^i\norm{(A \sminus \K C)^i}^2 
        \leq 1+\alpha^{-1},
    \end{align}
    which proves~\eqref{eq-uniform-stability} with $\cUS = \sqrt{1+\alpha^{-1}}$ and $\rhoUS = \frac{1}{\sqrt{1 + \alpha}} < 1$.
\end{proof}

\begin{lem}[Theorem 2B.1 in~\citet{Ljung1999SystemId}]\label{lem-LLN}
    Let $w_k$ be independent random variables with ${\, \sup\limits_{k \in \N} \, \Eof{\norm{w_k}^4} < \infty}$. \\
    Let $H(\theta)$ be a family of uniformly stable filters over $\theta \in \Theta$, i.e., for some $\lambda \in (0,1)$ and $c_H > 0$, ${\norm{H_i(\theta)} \le c_H \lambda^i}$ holds for all $i \in \N$ and $\theta \in \Theta$. \\
    Also define $s_k(\theta) \coloneqq \sum_{i=0}^{k} H_i(\theta) w_{k-i}$. \\
    Then, the following uniform law of large numbers holds almost surely (i.e., with probability one):
    \begin{align}
        \sup_{\theta \in \Theta} \frac{1}{N}\norm{\sum_{k=1}^N s_k(\theta)s_k(\theta)\transp - \Eof{s_k(\theta)s_k(\theta)\transp } } \mylimit{N} 0.
    \end{align}
\end{lem}

\begin{proof}[Proof of Lemma~\ref{lem-uniform-convergence}]
    The idea of the proof is to first show $\nabla V_N(\K) \to \Eof{\nabla V_N(\K)}$, then we will show $\Eof{\nabla V_N(\K)} \to \nabla \bar{V}(\K)$, and conclude from this that the desired result $\nabla V_N(\K) \to \nabla \bar{V}(\K)$ as in~\eqref{eq-uniform-convergence-dV} holds. \\
    First, $s_k(\theta)\equiv e_k - C z_k(\K)$ satisfies the conditions of Lemma~\ref{lem-LLN} with $\theta = \K$ and $\Theta = \Set$.
    For any pair of indices $(i,j)$, the same holds for $\frac{\partial z_k(\K)}{\partial \K_{ij}}$.
    Thus, applying this lemma leads to the first desired result:
    \begin{align}\label{eq-uniform-convergence-dV-1}
        \Pof{\sup_{\K \in \Set} \norm{ \frac{\partial V_N(\K)}{\partial \K_{ij}}  - \Eof{\frac{\partial V_N(\K)}{\partial \K_{ij}}} } \mylimit{N} 0}& =1.
    \end{align}

    Next, we express $\Eof{V_N(\K)}$ and compare it with $\bar{V}$ in~\eqref{eq-def-V-bar}:
    \begin{align}\label{eq-proof-uniform-convergence-1}
        \Eof{V_N(\K)}&=
        \Trace{W \ofc{ S^\star + C \of{\textstyle \frac{1}{N}\sum\limits_{k=1}^{N} \Covof{z_k(\K)}} C\transp} }
        \nonumber \\
        &= \bar{V}(\K)  + \Trace{ \! W C \of{\textstyle \frac{1}{N}\sum\limits_{k=1}^{N} X_k(\K)} C\transp},
    \end{align}
    with $X_k(\K) \coloneqq \Covof{z_k(\K)} - \bar{\Sigma}(\K)$.
    Using the the dynamics of $z_k(\K)$ in~\eqref{eq:zk} we can derive:
    \begin{align}
        \Covof{z_{k+1}(\K)} &= (A  - \K C) \Covof{z_k(\K)} (A  - \K C)\transp \\
        \nonumber &\hspace{20mm}  + (\K - \K^\star) S^\star  (\K - \K^\star)\transp.
    \end{align}
    Comparing this with the recursion for $\bar{\Sigma}(\K)$ in~\eqref{eq-def-Sigma-bar}, we find that $X_k(\K)$ satisfies:
    \begin{align}\label{eq-proof-uniform-convergence-2}
        X_{k+1}(\K) &= (A  - \K C) X_k(\K) (A  - \K C)\transp.
    \end{align}
    Using the uniform stability from Assumption~\ref{assum-K}, equation~\eqref{eq-proof-uniform-convergence-2} implies that $X_k(\K)$ converges to zero uniformly over $\K \in \Set$.
    This result, combined with~\eqref{eq-proof-uniform-convergence-1} leads to:
    \begin{align}\label{eq-uniform-convergence-dV-2}
        \sup_{\K \in \Set} \norm{ \Eof{\frac{\partial V_N(\K)}{\partial \K_{ij}}} - \frac{\partial \bar{V}(\K)}{\partial \K_{ij}} } \mylimit{N} 0.
    \end{align}
    By combining this with the uniform convergence equations~\eqref{eq-uniform-convergence-dV-1}, it follows that $\frac{\partial V_N(\K)}{\partial \K_{ij}} \to \frac{\partial \bar{V}(\K)}{\partial \K_{ij}}$ almost surely and uniformly.
    This holds for any pair of indices $(i,j)$ so the uniform convergence of the whole gradient is established.
\end{proof}

\end{document}